\documentclass[a4paper,10pt]{article}

\usepackage[pdftex]{graphicx}
\usepackage{latexsym}
\usepackage{amsmath,amssymb}
\usepackage{ascmac}
\usepackage[ruled,vlined,linesnumbered]{algorithm2e}
\usepackage{tikz}

\newcommand{\MX}{\mathcal{X}}
\newcommand{\MXall}{\mathcal{X}_G}
\newcommand{\MA}{\mathcal{A}}
\newcommand{\MC}{\mathcal{C}}
\newcommand{\MI}{\mathcal{I}}
\newcommand{\ML}{\mathcal{L}}
\newcommand{\MM}{\mathcal{M}}
\newcommand{\MT}{\mathcal{T}}
\newcommand{\MS}{\mathcal{S}}
\newcommand{\MY}{\mathcal{Y}}
\newcommand{\Null}{\textsc{Null}}

\newcommand{\Next}{\textsc{Next}}
\newcommand{\Prev}{\textsc{Prev}}

\newcommand{\algref}[1]{Algorithm~{\rm\ref{alg:#1}}}
\renewcommand{\eqref}[1]{(\ref{eq:#1})}
\newcommand{\figref}[1]{Figure~\ref{fig:#1}}
\newcommand{\lemref}[1]{Lemma~\ref{lem:#1}}
\newcommand{\lineref}[1]{line~\ref{line:#1}}

\renewcommand{\mid}{:\,}

\newcommand{\secref}[1]{Section~\ref{sec:#1}}

\newcommand{\thmref}[1]{Theorem~\ref{thm:#1}}
\long\def\invis#1{}

\newcommand{\True}{\textsc{True}}
\newcommand{\False}{\textsc{False}}

\newcommand{\odd}{\mathrm{odd}}
\newcommand{\even}{\mathrm{even}}
\newcommand{\astup}{{\ast\uparrow}}
\newcommand{\astdown}{{\ast\downarrow}}

\newtheorem{thm}{Theorem}

\newtheorem{lem}{Lemma}

\newenvironment{proof}{\medskip
  \noindent{\scshape Proof:}}{\quad $\Box$\medskip}

\newenvironment{myframe}{\begin{trivlist}\item[]
    \hrule
    \hbox to \linewidth\bgroup
    \advance\linewidth by -30pt
    \hsize=\linewidth
    \vrule\hfill
    \vbox\bgroup
    \vskip15pt
    
    \begin{minipage}{\linewidth}}{%
    \end{minipage}\vskip15pt
    \egroup\hfill\vrule
    \egroup\hrule
\end{trivlist}}
\begin{document}
\title{Maximum Weighted Matching with Few Edge Crossings for 2-Layered Bipartite Graph}

\invis{
\author{Kazuya Haraguchi\inst{1}\and Kotaro Torii\inst{2}\and Motomu Endo\inst{2}}
\authorrunning{K. Haraguchi et al.}
\institute{Otaru University of Commerce, Japan\\
  \email{haraguchi@res.otaru-uc.ac.jp}\and
  Nara Institute of Science and Technology}
}
\author{Kazuya Haraguchi\thanks{Corresponding author. E-mail: {\tt haraguchi@res.otaru-uc.ac.jp}}\and Kotaro Torii\and Motomu Endo}
\date{}


\maketitle

\begin{abstract}
  Let $c$ denote a non-negative constant. 
  Suppose that we are given an edge-weighted bipartite graph $G=(V,E)$
  with its 2-layered drawing and a family $\MX\subseteq E\times E$
  of intersecting edge pairs. 
  We consider the problem of finding a maximum weighted matching $M^\ast$
  such that each edge in $M^\ast$ intersects
  with at most $c$ other edges in $M^\ast$,
  and that all edge crossings in $M^\ast$ are contained in $\MX$.
  In the present paper, we propose polynomial-time algorithms for
  the cases of $c=1$ and 2.
  The time complexities of the algorithms are
  $O\big((k+m)\log n+n\big)$ for $c=1$ 
  and $O\big(k^3+k^2n+m(m+\log n)\big)$ for $c=2$, respectively,
  where $n=|V|$, $m=|E|$ and $k=|\MX|$. 
  %
  %
\end{abstract}

\section{Introduction}
\label{sec:intro}

Let $G=(A,B,E)$ denote an edge-weighted bipartite graph,
where $\{A,B\}$ is the bipartition of the entire vertex set $V=A\cup B$
and $E\subseteq A\times B$ is the edge set.
Denoting by $n_A=|A|$ and $n_B=|B|$,
we let $A=\{a_1,\dots,a_{n_A}\}$, $B=\{b_1,\dots,b_{n_B}\}$,
$n=n_A+n_B$ and $m=|E|$. 
We abbreviate $(a_i,b_q)\in E$ into $a_ib_q$ for simplicity. 
The edge weight is given by a function $w:E\rightarrow\mathbb{R}_+$. 
A subset $M\subseteq E$ of edges is called a {\em matching\/}
if no two edges in $M$ share an endpoint in common.
Denoted by $w(M)$, 
the weight of a matching $M$ is defined as the sum of edge weights over $M$,
i.e., $w(M)=\sum_{e\in M}w(e)$. 

A {\em $2$-layered drawing\/}~\cite{BETT.1999,K.2002,STT.1981} of a bipartite graph $G$
is a 2D drawing of $G$ such that 
$a_1,\dots,a_{n_A}\in A$
and $b_1,\dots,b_{n_B}\in B$ are put on two horizontal lines
as distinct points from left to right, respectively, and that
every edge is drawn as a straight line segment between the endpoints.
Two edges $e=a_ib_q$ and $e'=a_jb_p$
{\em make a crossing\/} or {\em intersect}
if either ($i<j$ and $p<q$) or ($i>j$ and $p>q$) holds.

In our research, we study the problem of computing a maximum weighted (max-weighted for short) matching under
a constraint such that only a small number of
edge crossings are admitted.
The constraint we take up here
is that each matching edge may intersect with at most $c$ other matching edges,
where $c$ is a non-negative constant. 

We formulate
the problem 
in a more general setting.  
For $e,e'\in E$, we call $\{e,e'\}$ a {\em crossing pair\/}
if $e$ and $e'$ make a crossing. 
Let $\MXall$ denote the set of all crossing pairs in $G$.
For input,
we accept a subset $\MX\subseteq \MXall$ as well as $G$ and $w$,
where $\MX$ is the set of crossing pairs
that are admitted to make crossings. 
Let us call $\MX$ an {\em admissible set\/}
and a crossing pair in $\MX$ an {\em admissible pair\/}.
A matching $M$ is called {\em at-most-$c$-crossings-per-edge}
({\em $c$-CPE}) if each edge in $M$ makes at most $c$ crossings
along with other edges in $M$ and every crossing pair that appears in $M$
belongs to $\MX$.

We formalize the {\em max-weighted $c$-CPE matching problem} ({\em MW-$c$-CPEMP}) as follows.

\begin{myframe}
  \begin{description}
  \item[\underline{Max-weighted $c$-CPE matching problem (MW-$c$-CPEMP)}]
  \item[Input:] A bipartite graph $G=(A,B,E)$ along with its 2-layered drawing,
    a positive edge weight function $w:E\rightarrow\mathbb{R}_+$,
    and an admissible set $\MX\subseteq\MXall$. 
  \item[Output:]  A $c$-CPE matching $M^\ast\subseteq E$ that maximizes $w(M^\ast)$. 
  \end{description}
\end{myframe}

For example, when $\MX=\MXall$,
we are asked to compute a max-weighted $c$-CPE matching
such that any crossing pair may appear. 
When $c=0$ or $\MX=\emptyset$,
the problem asks for a max-weighted non-crossing matching
since no crossing pair is admitted.

In the present paper,
we propose polynomial time algorithms 
for the MW-$c$-CPEMP with $c\in\{1,2\}$. 
Our approach reduces the MW-$c$-CPEMP to what we call the 
{\em non-contact trapezoid selection problem} ({\em NTSP}). 
We then solve the reduced NTSP problem
by an algorithm named \textsc{SelectTrape},
which is an extension of 
the Malucelli et al.'s $O(m\log n)$-time algorithm
for the MW-0-CPEMP~\cite{MOP.1993}. 
The time complexities of the proposed algorithms
are $O\big((k+m)\log n+n\big)$ for $c=1$
and $O\big(k^3+k^2n+m(m+\log n)\big)$ for $c=2$ respectively,
where $k=|\MX|$.

The paper is organized as follows.
We describe our motivation and related work in \secref{bg}.
In \secref{ntsp}, we introduce the NTSP and present 
the algorithm \textsc{SelectTrape}. 
We then explain how to reduce the MW-$c$-CPEMP
to the NTSP in \secref{cpemp},
followed by concluding remarks in \secref{conc}.

\section{Background}
\label{sec:bg}
\subsection{Motivation}

Plant chronobiologists
would like to compare gene expression dynamics at the individual level
(i.e., macro level) with the single cell level (i.e., micro level)
along the same time axis.
However, there is a technically hard issue.
Conventional microarray or RNA-sequencing can easily measure
individual gene expression patterns in actual time-series
but have limited spatial resolutions.
On the other hand, single-cell transcriptome techniques
have ultimate spatial resolution of gene expression analysis,
but most techniques are requiring destruction of cells
to perform single cell transcriptome and
thus actual time-series analysis is impossible.
Thus, to provide an analytic tool to
achieve higher spatiotemporal resolution was required. 

As an alternative, single cell analysis often
uses {\em pseudo time-series reconstruction}
for revealing cell-state transition (e.g., \cite{Wishbone,Monocle}).
Pseudo time reconstruction is a process
that orders cells transcriptome on a hypothetical time axis,
along which they show continuous changes in the transcriptome.
However, ordinal scale-based pseudo time-series will not
provide any time information so that
it is impossible to analyze circadian rhythm,
for example, in a single cell resolution.

We have hypothesized that timing of significant gene expression peak
on the pseudo time-series is comparable to that on the actual
time-series. In our recent work~\cite{Tetal},
we formulated the problem of
estimating the actual time of cell expressions
as the MW-0-CPEMP. 
We considered a 2-layered drawing
of a complete bipartite graph $G=(A,B,E)$
such that $A$ is the set of individual expression records
that are sorted in the actual-time order, and
$B$ is the set of cell expression records
that are sorted in a hypothetical order. 
We weighted each edge $a_ib_q\in E$ by a heuristic method
that evaluates how $a_i$ and $b_q$ are likely to match.
Solving the MW-0-CPEMP on $G$,
we estimated the actual time of a gene expression $b_q$
by the time of the individual expression $a_i$
to which $b_q$ is matched.
We imposed the non-crossing constraint to preserve the vertex orders. 
We observed that the model can be a useful tool for actual-time estimation,
compared with conventional actual-time estimation methods.

However, the non-crossing constraint is not necessarily a hard constraint
since cell expression records may contain noise in practice. 
In pursuit of alternative models, 
we study the MW-$c$-CPEMP for a constant $c\ge 1$.

\subsection{Related Work}
\label{sec:bg.rel}
The MW-0-CPEMP is an extension of the 
longest common subsequence problem
on given two sequences~\cite{CLRS.2009}.
It has an application
in the sequence alignment problem 
that appears in bioinformatics~\cite{WJ.1994} and  
in natural language processing~\cite{NSUHN.2016}.

Knauer et al.~\cite{KSSW.2007} studied the problem
of finding a subgraph that has few edge crossings;
given a graph (not necessarily bipartite)
and its geometric drawing (i.e., every vertex is specified by a 2D point, 
all edges are drawn as straight line segments,
and no two edges overlap or intersect at a vertex),
we are asked to find a subgraph of a certain class
that makes the minimum number of edge crossings.
They showed that,
for spanning trees, $s$-$t$ paths, cycles, matchings of a fixed size, and
1- or 2-factors,
it is NP-hard to approximate the minimum number of edge crossings
within a factor of $z^{1-\epsilon}$ for any $\epsilon>0$, where $z$ denotes the number of edge crossings in the given graph.
They also presented fixed-parameter algorithms
to decide whether there is a non-crossing subgraph
of one of the above graph classes,
where $z$ is used as the parameter.

The non-crossing (or crossing-free) constraint has been considered
for some problems of finding an ``optimal'' subgraph. 
It is Malucelli et al.~\cite{MOP.1993}
who first studied the algorithmic aspect of the MW-0-CPEMP
explicitly.
For the edge-unweighted case, 
they provided a polynomial-time algorithm
that runs in $O(m\log\log n)$ time or in $O(m+\min\{n\mu,m\log\mu\})$ time,
where $\mu$ denotes the cardinality of a maximum 0-CPE matching. 
They also extended the algorithm to the edge-weighted case,
which yields an $O(m\log n)$ time algorithm. 
A bipartite graph is {\em convex} if, for every $a_i\in A$,
$a_ib_p,a_ib_q\in E$ $(p\le q)$ implies
$a_ib_{p'}\in E$ for all $p\le p'\le q$.
For the MW-0-CPEMP in
edge-unweighted convex bipartite graphs,
Chen et al.~\cite{CLW.2015} presented an algorithm
whose running time is $O(n\log n)$. 
Carlsson et al.~\cite{CARB.2015} considered
the Euclidean non-crossing bipartite matching problem,
where each vertex is represented by a 2D point.
The objective is to find a non-crossing
perfect matching whose longest edge is minimized. 
They showed that the problem is NP-hard in general,
but that it is polynomially-solvable in some special cases. 
More recently, Altinel et al.~\cite{AAST.2018}
showed that the minimum cost non-crossing flow problem
on a layered network is NP-hard.
Ruangwises and Itoh~\cite{RI.2019}
studied the stable marriage problem
under the non-crossing constraint,
showing that there exists a weakly stable non-crossing
matching for any instance.

The {\em conflict pair constraint} 
(or {\em negative disjunctive constraint})
is a generalization of the non-crossing constraint. 
Represented by a {\em conflict graph} $\hat{G}=(E,\MC)$,
this constraint prohibits a solution from 
including two edges $e,e'\in E$ such that $\{e,e'\}\in \MC$. 
The minimum cost perfect matching problem with conflict pair constraints
is strong NP-hard for a general graph $G$ even if the conflict graph $\hat{G}$
is a collection of single edges~\cite{DPSW.2011}
and is NP-hard even for $G$ that consists of 4-cycles~\cite{OZP.2013};
it turns out that the problem is NP-hard for a bipartite graph.
For this type of constraint,
there are also studies on the transportation problem~\cite{S.1998,S.2002},
the minimum spanning tree problem~\cite{DPS.2009,SU.2015,ZKP.2011},
and the max-flow problem~\cite{PS.2013}.

\section{Non-contact Trapezoid Selection Problem}
\label{sec:ntsp}

To solve the MW-$c$-CPEMP,
we reduce the problem to  
what we call the non-contact trapezoid selection
problem (NTSP). 
In this section, we define the NTSP 
and propose an efficient algorithm for it.
The algorithm is an extension of the Malucelli et al.'s algorithm~\cite{MOP.1993} for the MW-0-CPEMP.  

\subsection{Problem Description}
Suppose two distinct horizontal lines on the 2D plane.
Let $A=\{a_1,\dots,a_{n_A}\}$ and $B=\{b_1,\dots,b_{n_B}\}$
denote vertex sets. 
We put $a_1,\dots,a_{n_A}$ on the upper line from left to right,
and $b_1,\dots,b_{n_B}$ on the lower line from left to right.  
We are given a collection $\MT=\{T_1,\dots,T_z\}$
of weighted trapezoids such that
each $T_s\in\MT$ is given its weight,
denoted by $\omega(T_s)$,
and its two upper corners are among $A$,
whereas its two lower corners are among $B$. 
We denote by $\lambda_A(T_s)$ (resp., $\gamma_A(T_s)$)
the index $i$ (resp., $j$) of 
the upper-left corner $a_i$ (resp., upper-right corner $a_j$).
Similarly, we denote by $\lambda_B(T_s)$
(resp., $\gamma_B(T_s)$)
the index $p$ (resp., $q$) of 
the lower-left corner $b_p$ (resp., lower-right corner $b_q$). 
We admit $T_s$ to be a triangle or a line segment. 
Then $\lambda_A(T_s)\le\gamma_A(T_s)$ and
$\lambda_B(T_s)\le\gamma_B(T_s)$ hold.

Given $(A,B,\MT,\omega)$,
the NTSP asks for a max-weighted subcollection $\MS\subseteq\MT$
such that any $T_s,T_t\in\MS$ do not contact each other. 
Specifically, for any $T_s,T_t\in\MS$ with $T_s\ne T_t$
and $\lambda_A(T_s)\le\lambda_A(T_t)$,
it should hold that $\gamma_A(T_s)<\lambda_A(T_t)$ and $\gamma_B(T_s)<\lambda_B(T_t)$. 


\subsection{Partial Order Based Algorithm}
\label{sec:ntsp.poset}
We can solve the NTSP by using the notion of partial order. 
Let us introduce a binary relation $\prec$ on $\MT$;
For $T_s,T_t\in\MM$, we write $T_s\prec T_t$
if 
$\gamma_A(T_s)<\lambda_A(T_t)$
and $\gamma_B(T_s)<\lambda_B(T_t)$. 
One easily sees that $\prec$ is a (or an irreflexive) partial order on $\MT$,
and thus $(\MT,\prec)$ is a partially ordered set (poset). 
We say that $T_s$ and $T_t$ are {\em comparable\/}
if either $T_s\prec T_t$ or $T_t\prec T_s$ holds.
For a subcollection $\MC\subseteq\MT$,
the poset $(\MC,\prec)$ (or $\MC$) is called a {\em chain}
if every $T,T'\in\MC$ are comparable. 
Obviously $\MC$ is a feasible solution of the NTSP
iff it is a chain. 

We represent the poset $(\MT,\prec)$ by a 
directed acyclic graph (DAG).
We denote the DAG by $D=(\MT\cup\{\phi\},\MA)$,
where $\phi$ is the dummy node
and $\MA=\MA^\phi\cup \MA^\prec$
is the arc set such that
\begin{align*}
  \MA^\phi&=\{(\phi,T')\mid T'\in\MT\},\\
  \MA^\prec&=\{(T,T')\in\MT\times\MT\mid T\prec T'\}. 
\end{align*}
For an arc $(T,T')\in \MA$, 
we define the distance to be $\omega(T')$. 
Any feasible solution of the NTSP is represented by a path from $\phi$. 
Then we can solve the NTSP by solving the longest path problem on $D$.
The time complexity of this algorithm 
is $O(z^2)$ because we can construct $D$ 
and solve the longest path problem in $O(z^2)$ time~\cite{CLRS.2009}.

\subsection{An Efficient Algorithm {\sc SelectTrape}}
We propose a faster algorithm whose time complexity is $O(z\log n+n)$. 
The proposed algorithm is based on the 
Malucelli et al.'s algorithm~\cite{MOP.1993}
for the MW-0-CPEMP. 
The MW-0-CPEMP is regarded as a special case
of the NTSP such that every trapezoid $T\in\MT$
is a line segment, that is, $\lambda_A(T)=\gamma_A(T)$ and
$\lambda_B(T)=\gamma_B(T)$.
We extend the Malucelli et al.'s algorithm 
based on this observation.

For two integers $i,j$
with $1\le i\le j\le n_A$,
let us denote $[i,j]=\{i,\dots,j\}$ and $[i]=[1,i]$.  
We define $A[i,j]\triangleq\{a_i,\dots,a_j\}$
and $A[j]\triangleq A[1,j]$. 
Similarly, for integers $p,q$ with $1\le p\le q\le n_B$,
we define a subset $B[p,q]\triangleq\{b_p,\dots,b_q\}$
and $B[q]\triangleq B[1,q]$. 
We say that a trapezoid $T$ is {\em contained in $A[i,j]\cup B[p,q]$}
if all corners are contained in the vertex subset, that is,
$\lambda_A(T),\gamma_A(T)\in [i,j]$ and
$\lambda_B(T),\gamma_B(T)\in [p,q]$. 

For $T_s\in\MT$, 
we denote by $\mu_\MT(T_s)$ the max-weight
of a feasible solution that has $T_s$ as the rightmost trapezoid
(i.e., no $T$ in the solution satisfies $T_s\prec T$).
For $(i,q)\in [n_A]\times[n_B]$, we denote by $\mu_\gamma(i,q)$
the max-weight of a feasible solution
such that the trapezoids are contained in $A[i]\cup B[q]$
and the lower-right corner of the rightmost trapezoid is exactly $b_q$.
For convenience, we let $\mu_\gamma(0,q)=0$ for all $q\in[n_B]$. 
The following lemmas are obvious by the definitions.

\begin{lem}
  \label{lem:NTSP.muMT}
  Suppose that $(A,B,\MT,\omega)$ is given.
  For $T_s\in\MT$, we have
  \begin{align}
    \mu_\MT(T_s)=\omega(T_s)+\max\Big\{\mu_\gamma(\lambda_A(T_s)-1,q)\mid q\in[\lambda_B(T_s)-1]\Big\}.\label{eq:NTSP.muMT}
  \end{align}
\end{lem}

\begin{lem}
  \label{lem:NTSP.mugamma}
  Suppose that $(A,B,\MT,\omega)$ is given.
  For $(i,q)\in [n_A]\times[n_B]$, we have  
  \begin{align*}
  \mu_\gamma(i,q)=
  \max\Big\{\mu_\gamma(i-1,q),\max_{T_s\in\MT\mid (\gamma_A(T_s),\gamma_B(T_s))=(i,q)}\{\mu_\MT(T_s)\}\Big\}.
  \end{align*}
\end{lem}

In \algref{NTSP}, 
we show an algorithm \textsc{SelectTrape}
that computes the optimal weight 
of a given NTSP instance.
The algorithm repeats the outer for-loop for $i=1,\dots,n_A$.
We do not store $\mu_\gamma(i,q)$ for all $i\in[n_A]$,
but only for the current $i$.
It is stored as $\hat{\mu}_\gamma(q)$. 
When every $T\in\MT$ is a line segment
(i.e., $\lambda_A(T)=\gamma_A(T)$ and $\lambda_B(T)=\gamma_B(T)$), 
the algorithm works exactly in the same way as the Malucelli et al.'s algorithm. 

\begin{algorithm}[t!]
  \caption{An algorithm \textsc{SelectTrape}
    to compute the optimal value of a given NTSP instance}
  \label{alg:NTSP}
  \DontPrintSemicolon
  \SetKwInOut{Input}{Input}\SetKwInOut{Output}{Output}
  \Input{An instance $(A,B,\MT,\omega)$, where $A=\{a_1,\dots,a_{n_A}\}$, $B=\{b_1,\dots,b_{n_B}\}$,
    $\MT=\{T_1,\dots,T_z\}$, and $\omega:\MT\rightarrow\mathbb{R}_+$}
  \Output{The max-weight of a subcollection
    $\MS\subseteq \MT$
    such that every $M_s,M_t\in\MS$ $(M_s\ne M_t)$ do not contact each other}
  $\hat{\mu}_\gamma(q)\gets 0$ for all $q\in[n_B]$\;
  \label{line:NTSP.pre}
  $\omega^\ast\gets-\infty$\;
  \For{$i=1$ to $n_A$}{
    \label{line:NTSP.forbegin}
    \For{$T_s\in\MT$ such that $\lambda_A(T_s)=i$}{
      \label{line:NTSP.l_forbegin}
      $p\gets\lambda_B(T_s)$\;
      $\mu_\MT(T_s)\gets \omega(T_s)+\max\{\hat{\mu}_\gamma(q)\mid q\in [p-1]\}$\;
      \label{line:NTSP.muMT}
    }
    \For{$T_t\in\MT$ such that $\gamma_A(T_t)=i$}{
      \label{line:NTSP.inforbegin}
      $q\gets\gamma_B(T_t)$\;
      $\hat{\mu}_\gamma(q)\gets \max\{\hat{\mu}_\gamma(q),\mu_\MT(T_t)\}$\;
      \label{line:NTSP.max}
      \lIf{$\hat{\mu}_\gamma(q)>\omega^\ast$}{$\omega^\ast\gets\hat{\mu}_\gamma(q)$}
    }
    \label{line:NTSP.inforend}
  }
  {\bf output} $\omega^\ast$
\end{algorithm}

\begin{thm}
  \label{thm:NTSP}
  Given an instance $(A,B,\MT,\omega)$ of the NTSP, 
  the algorithm {\sc SelectTrape} in \algref{NTSP}
  computes the optimal weight in $O(z\log n+n)$ time and in $O(z+n)$ space,
  where $z=|\MT|$. 
\end{thm}
\begin{proof}
  We show that $\hat{\mu}_\gamma(q)=\mu_\gamma(i-1,q)$ holds
  at the beginning of the outer for-loop with respect to $i$
  (i.e., \lineref{NTSP.forbegin}). 
  When $i=1$, we have $\hat{\mu}_\gamma(q)=0=\mu_\gamma(0,q)$ by
  \lineref{NTSP.pre}.
  In \lineref{NTSP.muMT},
  we see that $\mu_\MT(T_s)$ is computed correctly
  for $T_s\in\MT$ with $\lambda_A(T_s)=i$,
  due to $\hat{\mu}_\gamma(q)=\mu_\gamma(i-1,q)=\mu_\gamma(\lambda_A(T_s)-1,q)$ (by induction)
  and \lemref{NTSP.muMT}.
  The second inner for-loop (i.e., \lineref{NTSP.inforbegin}
  to \ref{line:NTSP.inforend}) computes the maximum
  among $\hat{\mu}_\gamma(q)=\mu_\gamma(i-1,q)$
  and $\mu_\MT(T_s)$ for all $T_s\in\MT$ with $(\gamma_A(T_s),\gamma_B(T_s))=(i,q)$,
  and substitutes the maximum for $\hat{\mu}_\gamma(q)$,
  which is $\mu_\gamma(i,q)$ by \lemref{NTSP.mugamma}. 
  Upon completion of the algorithm,
  we have 
  $\omega^\ast=\max_{q\in[n_B]}\{\hat{\mu}_\gamma(q)\}=\max_{q\in[n_B]}\{\mu_\gamma(n_A,q)\}$,
  which is the optimal value. 
  
  We analyze the computational complexity.
  Each trapezoid is searched as $T_s$ in the first inner for-loop exactly once,
  and as $T_t$ in the second inner for-loop exactly once.
  We can access $T_s\in\MT$ with $\lambda_A(T_s)=i$
  in \lineref{NTSP.l_forbegin}
  (and $T_t\in\MT$ with $\gamma_A(T_t)=i$ in \lineref{NTSP.inforbegin})
  in $O(1)$ time by executing the bucket sort on $\MT$ beforehand.
  The bucket sort runs in $O(z+n)$ time. 
  We use {\em priority search tree}~\cite{M.1985}
  to store $(q,\hat{\mu}_\gamma(q))$ for $q\in[n_B]$,
  by which we can take the maximum in \lineref{NTSP.max} in $O(\log n_B)$ time.
  We see that the algorithm runs in $O(z\log n+n)$ time. 
  We use $O(z+n)$ space to store $\mu_\MT(T_s)$ for $T_s\in\MT$
  and $\hat{\mu}_\gamma(q)$ for $q\in[n_B]$. 
\end{proof}

We explain how to construct an optimal solution $\MS^\ast$. 
For the rightmost trapezoid $T_t$ in $\MS^\ast$,
it holds that $\mu_\MT(T_t)=\sum_{T\in\MS^\ast}\omega(T)=\omega^\ast$.
Similarly, for $\MS=\MS^\ast\setminus\{T_t\}$
and the rightmost trapezoid $T_{t'}$ in $\MS$,
since $\MS$ is a max-weighted feasible solution among those having $T_{t'}$ as the rightmost trapezoid,
it holds that $\mu_\MT(T_{t'})=\sum_{T\in\MS}\omega(T)=\omega^\ast-\omega(T_t)$.
Note that, since $T_{t'}\prec T_t$, $\gamma_A(T_{t'})<\lambda_A(T_t)$ and $\gamma_B(T_{t'})<\lambda_B(T_t)$ should hold. 
Then we can construct $\MS^\ast$ by \algref{NTSPconst} after running \algref{NTSP}. 
The running time is $O(z+n)$.

\begin{algorithm}[t!]
  \caption{An algorithm to construct an optimal solution of a given NTSP instance}
  \label{alg:NTSPconst}
  \DontPrintSemicolon
  \SetKwInOut{Input}{Input}\SetKwInOut{Output}{Output}
  \Input{An instance $(A,B,\MT,\omega)$, where $A=\{a_1,\dots,a_{n_A}\}$, $B=\{b_1,\dots,b_{n_B}\}$,
    $\MT=\{T_1,\dots,T_z\}$, and $\omega:\MT\rightarrow\mathbb{R}_+$,
    the optimal weight $\omega^\ast$ of the instance, and
    a function $\mu_\MT:\MT\rightarrow\mathbb{R}_+$ the satisfies \eqref{NTSP.muMT}
  }
  \Output{An optimal solution}
  $\MS\gets\emptyset$;
  $\alpha\gets\omega^\ast$;
  $i\gets n_A$;
  $q\gets n_B$\;
  \While{$\alpha>0$}{
    \For{$T_t\in\MT$ such that $\gamma_A(T_t)=i$ and $\gamma_B(T_t)\le q$}{
      \If{$\mu_\MT(T_t)=\alpha$}{
        $\MS\gets\MS\cup\{T_t\}$;
        $\alpha\gets\alpha-\omega(T_t)$;
        $i\gets\lambda_A(T_t)$;
        $q\gets\lambda_B(T_t)-1$\;
        {\bf break} the for-loop\;
      }
    }
    $i\gets i-1$\;
  }
  {\bf output} $\MS$\;
\end{algorithm}

\section{Algorithm for the MW-$c$-CPEMP}
\label{sec:cpemp}

In this section, we reduce the MW-$c$-CPEMP to the NTSP,
which yields polynomial-time algorithms for the cases of $c=1$ and 2.
Throughout this section, we assume that an instance $(G,w,\MX)$ of the MW-$c$-CPEMP is given
for a non-negative constant $c$.
We also assume that $\MX$ is expressed by
a list of crossing pairs.

  Before proceeding, let us mention that the problem is solvable in $O^\ast(2^k)$-time as follows,
  where $k=|\MX|$ and $O^\ast(\cdot)$ is introduced to ignore polynomial factors;
  For each $\MY\subseteq\MX$, let $M_\MY=\{e\in E\mid \exists Y\in\MY,\ e\in Y\}$.
  In other words, $M_\MY$ is the set of edges that appear in $\MY$. 
  We check whether $M_\MY$ is a $c$-CPE matching,
  which can be done in polynomial time.
  If it is the case, we compute a max-weighted non-crossing matching, say $M'_\MY$, on the subgraph that
  is obtained by removing $M_\MY$ and the extreme points from $G$.
  The $M'_\MY$ can be computed in polynomial time
  by using the Malucelli et al's algorithm~\cite{MOP.1993}. 
  The max-weighted $c$-CPE matching $M_\MY\cup M'_\MY$ over $\MY\subseteq\MX$ is an optimal solution.

\subsection{Notations}
For $V'\subseteq V$,
we define $\lambda_A(V')$ to be the smallest index
among the vertices in $A\cap V'$.
For convenience, we let $\lambda_A(V')$ be zero when $A\cap V'=\emptyset$.
That is, 
\begin{align*}
  \lambda_A(V')\triangleq
  \left\{
  \begin{array}{ll}
    \min_{a_i\in A\cap V'}\{i\} & \textrm{if\ }A\cap V'\ne\emptyset,\\
    0 &\textrm{otherwise.}
  \end{array}
  \right.
\end{align*}
Similarly, we define $\gamma_A(V')$ to be the largest index
among the vertices in $A\cap V'$, and for convenience,
it is set to $n_A+1$ if $A\cap V'=\emptyset$.
That is,
\begin{align*}
  \gamma_A(V')\triangleq\left\{
  \begin{array}{ll}
    \max_{a_j\in A\cap V'}\{j\} & \textrm{if\ }A\cap V'\ne\emptyset,\\
    n_A+1 & \textrm{otherwise.}
  \end{array}
  \right.
\end{align*}
Observe that, when $A\cap V'\ne\emptyset$, $\lambda_A(V')$ and $\gamma_A(V')$ represent
the indices of the leftmost and rightmost vertices in $A\cap V'$, respectively. 
We define $\lambda_B(V')$ and $\gamma_B(V')$ in the analogous way.

For an edge subset $E'\subseteq E$,
we define $V[E']$ to be the set of extreme points of edges in $E'$.
For simplicity,
We write $\lambda_A(V[E'])$ by $\lambda_A(E')$.
The notations $\gamma_A(E')$, $\lambda_B(E')$ and $\gamma_B(E')$ are analogous. 
When $E'$ is a singleton, that is,
$E'=\{e\}$ for some edge $e\in E$,
we write $\lambda_A(\{e\})$ as $\lambda_A(e)$.
In this case, it holds that $\lambda_A(e)=\gamma_A(e)$ and
$\lambda_B(e)=\gamma_B(e)$. 
We write an inequality $\lambda_A(e)\le\lambda_A(e')$
as $e\le_Ae'$. 
Using this notation, $e$ and $e'$ intersect
if ($e<_Ae'$ and $e'<_Be$) or ($e'<_Ae$ and $e<_Be'$). 

Recall that $\MXall$ denotes the set of all possible crossing pairs in $G$. 
Observe that each $X\in\MXall$
is a matching that consists of two intersecting edges.
We may write $X=\{e,e'\}$ in $\MXall$ as an ordered pair $X=(e,e')$
when we assume $e<_Ae'$ (and thus $e'<_Be$). 
For a matching $M\subseteq E$,
we define $\MXall[M]$ to be a set of crossing pairs
that appear in $M$, that is,
\[
\MXall[M]\triangleq\big\{\{e,e'\}\in \MXall\mid e,e'\in M\big\}.
\]
A matching $M$ is {\em at-most-$c$-crossings-per-edge} ({\em $c$-CPE})
if $\MXall[M]\subseteq\MX$ holds
and each $e\in M$ appears in at most $c$ crossing pairs in $\MXall[M]$.
Hence, $M$ is 0-CPE iff $\MXall[M]=\emptyset$.

\subsection{Overview}
Let $\MM$ denote the family of all matchings in $G$. 
We regard any $M\in\MM$ 
as a trapezoid $T_M$ that has
$a_i$ with $i=\lambda_A(M)$ as the upper-left corner,
$a_j$ with $j=\gamma_A(M)$ as the upper-right corner,
$b_p$ with $p=\lambda_B(M)$ as the lower-left corner, and
$b_q$ with $q=\gamma_B(M)$ as the lower-right corner,
and that has the weight $\omega(T_M)=\sum_{e\in M}w(e)$. 
Then $(\MM,\prec)$ is a poset, 
where $\prec$ is the partial order on a trapezoid collection
that we introduced in \secref{ntsp.poset}.

\begin{lem}
  \label{lem:either}
  For 
  a $2$-layered bipartite graph $G$, 
  let $M\in\MM$. 
  For any $e,e'\in M$, exactly one of the following holds:
  \begin{description}
  \item[(i)] $\{e,e'\}\in\MXall$; and
  \item[(ii)] $e$ and $e'$ are comparable.
  \end{description}
\end{lem}
\begin{proof}
  When $e=e'$, (ii) holds. Suppose that $e\ne e'$. 
  If $e$ and $e'$ intersect, then we have $\{e,e'\}\in\MXall$. 
  Otherwise, either $e\prec e'$ or $e'\prec e$ should hold
  as they do not share endpoints in common. 
\end{proof}

We introduce an auxiliary graph, which we denote by $H=(E,\MXall)$.
We call an edge in the underlying graph $G$ a {\em node}
when we use it in the context of $H$.
For $M\in\MM$, 
we denote by $H[M]=(M,\MXall[M])$ the subgraph induced by $M$.
Let $\MC_M$ denote the family 
of connected components in $H[M]$. 

\begin{lem}
  \label{lem:comp}
  For 
  a $2$-layered bipartite graph $G$, 
  let $M\in\MM$. Then $(\MC_M,\prec)$ is a chain. 
\end{lem}
\begin{proof}
  We show that any $X,Y\in\MC_M$ $(X\ne Y)$ are comparable. 
  %
  The $X$ and $Y$ are node sets of connected components of $H[M]$.
  For any $e_X\in X$ and $e_Y\in Y$, $e_X$ and $e_Y$ are not adjacent. 
  Then $\{e_X,e_Y\}\notin\MXall$ holds. 
  By \lemref{either}, $e_X$ and $e_Y$ are comparable.
  We assume that $e_X\prec e_Y$ without loss of generality.
  Suppose that there is $e'_Y\in Y$
  such that $e'_Y\prec e_X$.
  Since $H[Y]$ is connected, 
  there is a path between $e_Y$ and $e'_Y$.
  The path should contain an edge $e''_Y$
  that intersects with $e_X$ in $G$.
  Then $\{e_X,e''_Y\}\in\MXall$, which contradicts that $e_X$ and $e''_Y$ are not adjacent in $H$.  
  We see that $e_X\prec e_{Y'}$ holds
  for any $e_X\in X$ and $e_{Y'}\in Y$
  and thus $X\prec Y$ holds.  
\end{proof}

Let $\MM_c\subseteq\MM$ denote the family of all $c$-CPE matchings. 
For a $c$-CPE matching $M\in\MM_c$,
it holds that $\MXall[M]\subseteq\MX$
and the degree of any node in $H[M]$ is at most $c$. 
We call $M$ {\em connected\/} if $H[M]$ is connected. 
By \lemref{comp}, any $M\in\MM_c$ is partitioned into connected
$c$-CPE matchings.

If we are given a family $\MT$ of all connected
$c$-CPE matchings
(which are regarded as trapezoids),
then we can find a max-weighted $c$-CPE matching
by solving the corresponding NTSP. 
The time complexity of this algorithm is $O(\Gamma+|\MT|\log n+n)$ by \thmref{NTSP},
where $\Gamma$ denotes the time for constructing $\MT$.  
Then, if $\Gamma$ and $|\MT|$ are polynomially bounded,
the total running time of the algorithm is also polynomially bounded. 

Let us consider how to construct $\MT$.
For $M\in\MM_c$,
the degree of any node in $H[M]$ is at most $c$.
Then, when $c=0$, we have $\MT=E$, that is,
the collection of isolated nodes in $H$.
Then $|\MT|=m$ holds. 
When $c=1$, we have $\MT=E\cup\MX$, and thus $|\MT|=m+k$ holds.
For $c\in\{0,1\}$, 
the following theorems are immediate. 

\begin{thm}[Malucelli et al.~\cite{MOP.1993}]
  \label{thm:CPEMP.0}
  Given an instance $(G,w)$ of the MW-$0$-CPEMP,
  we can find a max-weighted $0$-CPE matching
  in $O(m\log n)$ time
  and in $O(m+n)$ space. 
\end{thm}

\begin{thm}
  \label{thm:CPEMP.1}
  Given an instance $(G,w,\MX)$ of the MW-$1$-CPEMP,
  we can find a max-weighted $1$-CPE matching
  in $O\big((k+m)\log n+n\big)$ time
  and in $O(k+m+n)$ space. 
\end{thm}

For \thmref{CPEMP.0},
the complexity would be $O(m\log n+n)$
if we apply \thmref{NTSP} directly to the analysis.
However, the second term $O(n)$ can be dropped
since we do not need to execute the bucket sort before running \algref{NTSP}.
When $c=0$, each trapezoid is an edge itself.
The bucket sort is not needed on condition that
the graph is represented by an adjacency list.

\subsection{Trapezoid Collection for $c=2$}
\label{sec:prop}

For a connected 2-CPE matching $M$,
the auxiliary graph $H[M]$ is an isolated point,
an edge, a cycle or a path, and
$\MXall[M]\subseteq\MX$ should hold.  
The next lemma tells that, when it is a cycle, the length (which is $|M|$)
is at most four. We call a cycle an {\em $\ell$-cycle} if the length is $\ell$. 

\begin{lem}
  \label{lem:cycle}
  For a $2$-layered bipartite graph $G$
  and an admissible set $\MX$,
  let $M$ be a $2$-CPE matching.
  If $H[M]$ is a cycle,
  then $|M|$ is at most four. 
\end{lem}
\begin{proof}
  Each edge in $M$ intersects with exactly two other edges in $M$.
  Hence $|M|\ge3$.
  Let $M=\{e_1,\dots,e_d\}$ $(d\ge3)$. 
  Without loss of generality,
  we suppose that $e_{1}<_Ae_{t}$ holds for all $t\in[2,d]$,
  that $e_2$ and $e_3$ are two edges that intersect with $e_1$,
  and that $e_2<_Ae_3$ holds. 

  As shown in \figref{cycle}~(a),
  if $e_2$ and $e_3$ intersect, then
  we see that any of $\{e_1,e_2,e_3\}$ intersects with two others.
  Since $H[M]$ is a cycle,
  $M=\{e_1,e_2,e_3\}$ holds. 
  Otherwise (i.e., if $e_2$ and $e_3$ do not intersect),
  as shown in \figref{cycle}~(b),
  we have $e_1<_Ae_2<_Ae_3$ and $e_2<_Be_3<_Be_1$. 
  Since $H[M]$ is a cycle, 
  the edge $e_2$ intersects with another edge in $M$, say $e_4$.
  From the definition of $e_1$, we have $e_1<_Ae_4$. 
  Since $e_1$ and $e_4$ should not intersect,
  $e_2<_Be_3<_Be_1<_Be_4$ holds.
  Since $e_2$ and $e_4$ intersect,
  $e_1<_Ae_4<_Ae_2<_Ae_3$ holds.
  We see that $e_3$ intersects with $e_1$ and $e_4$,
  and $e_4$ intersects with $e_2$ and $e_3$.
  Any of $\{e_1,\dots,e_4\}$ intersects with two others,
  and thus we have $M=\{e_1,\dots,e_4\}$. 
\end{proof}

\begin{figure}[t!]
  \centering
  \begin{tabular}{clc}
    \begin{tikzpicture}
      \node [circle,draw,fill=white] (i1) at (0,2) {}; 
      \node [circle,draw,fill=white] (i2) at (1.5,2) {}; 
      \node [circle,draw,fill=white] (i3) at (3,2) {}; 
      \node [circle,draw,fill=white] (q3) at (0,0) {}; 
      \node [circle,draw,fill=white] (q2) at (2,0) {}; 
      \node [circle,draw,fill=white] (q1) at (3,0) {}; 
      \draw (i1) -- (q1) node [left,pos=0.25] {$e_1$};
      \draw (i2) -- (q2) node [left,pos=0.25] {$e_2$};
      \draw (i3) -- (q3) node [right,pos=0.25] {$e_3$};
    \end{tikzpicture}
    &\ \ \ &
    \begin{tikzpicture}
      \node [circle,draw,fill=white] (i1) at (0,2) {}; 
      \node [circle,draw,fill=white] (i4) at (1,2) {}; 
      \node [circle,draw,fill=white] (i2) at (2,2) {}; 
      \node [circle,draw,fill=white] (i3) at (3,2) {}; 
      \node [circle,draw,fill=white] (q2) at (0,0) {}; 
      \node [circle,draw,fill=white] (q3) at (1,0) {}; 
      \node [circle,draw,fill=white] (q1) at (2,0) {}; 
      \node [circle,draw,fill=white] (q4) at (3,0) {}; 
      \draw (i1) -- (q1) node [left,pos=0.3] {$e_1$};
      \draw (i2) -- (q2) node [right,pos=0.15] {$e_2$};
      \draw (i3) -- (q3) node [right,pos=0.3] {$e_3$};
      \draw (i4) -- (q4) node [left,pos=0.15] {$e_4$};
    \end{tikzpicture}\\
    (a) 3-cycle && (b) 4-cycle
  \end{tabular}
  \caption{2-CPE matchings that appear as cycles in the auxiliary graph $H$}
  \label{fig:cycle}
\end{figure}
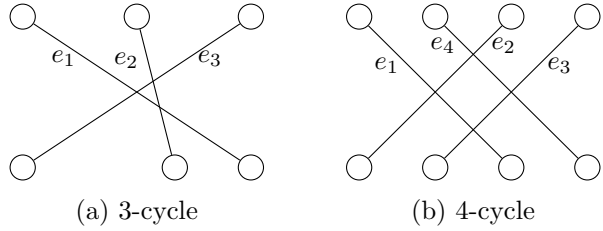

\begin{lem}
  \label{lem:enumcycle}
  For a $2$-layered bipartite graph $G$
  and an admissible set $\MX$,
  we can enumerate all $3$- and $4$-cycles
  in $O(k^2+m^2)$ time and in $O(m^2)$ space. 
\end{lem}
\begin{proof}
  First, we construct an $m\times m$ intersection matrix
  $\MI$ such that
  each row/column corresponds to an edge,
  and that each entry takes 1 (resp., 0)
  if the corresponding edge pair belongs to $\MX$
  (resp., does not belong to $\MX$).
  We can construct $\MI$ in $O(m^2)$ time
  and store it in $O(m^2)$ space.

  We can enumerate all 4-cycles in $O(k^2)$ time as follows; for each $X,Y\in\MX$,
  let $X=(e_1,e_3)$ and $Y=(e_4,e_2)$,
  where we assume $e_1\le_Ae_4$ without loss of generality. 
  We check whether $X\cup Y$ is a matching such that
  $\{e_1,e_2\},\{e_3,e_4\}\in\MX$ and $\{e_1,e_4\},\{e_2,e_3\}\notin\MX_G$.
  If yes, then $X\cup Y$ is a 4-cycle (\figref{cycle}~(b)).
  The check can be done in $O(1)$ time since
  whether $\{e,e'\}\in\MX$ or not can be identified in $O(1)$ time by using $\MI$. 
  Enumeration of 3-cycles is analogous. 
\end{proof}

There may exist an exponentially large number of paths in $H$. 
However, for our purpose, it is sufficient to take
into account only $O(k^2)$ paths;
Let $M$ be a 2-CPE matching such that $H[M]$ is a path. 
There are two nodes $e,e'\in M$ whose degrees are one. 
This means that $e$ and $e'$ appear in
exactly one admissible pair in $\MXall[M]$.
Suppose that $e<_Ae'$ holds without loss of generality. 
We call $X\in\MXall[M]$ with $e\in X$ (resp., $e'\in X$)
the {\em leftmost} (resp., {\em rightmost}) {\em admissible pair of $M$}. 
For $X,Y\in\MX$, we call a 2-CPE matching $M$ an {\em $(X,Y)$-path}
if $H[M]$ is a path and $X$ and $Y$ are the leftmost and rightmost admissible pairs of $M$, respectively. 
Among all $(X,Y)$-paths,
we have only to take a max-weighted one into account
because all $(X,Y)$-paths
form the same trapezoid whose corners
are $a_i$, $a_j$, $b_p$, and $b_q$,
where $i=\lambda_A(X)$, $j=\gamma_A(Y)$,
$p=\lambda_B(X)$, and $q=\gamma_B(Y)$. 
We define the size of an $(X,Y)$-path $M$ to be $|M|$,
that is, the number of edges in the matching $M$. 
In \figref{path}~(a) and (b),
we show a 2-CPE matching $M=\{e_1,\dots,e_8\}$ that is an $(X,Y)$-path
for $X=\{e_1,e_2\}$ and $Y=\{e_7,e_8\}$. 
The size of $M$ is eight.
We also show the path $H[M]$ in the auxiliary graph $H$ in \figref{HM}. 

\begin{figure}[t!]
  \centering
  \begin{tabular}{clc}
    \begin{tikzpicture}
      \node [circle,draw,fill=white] (i1) at (0,2) {};
      \node [circle,draw,fill=white] (i3) at (0.4,2) {};
      \node [circle,draw,fill=white] (i2) at (1.2,2) {};
      \node [circle,draw,fill=white] (i5) at (1.6,2) {};
      \node [circle,draw,fill=white] (i4) at (2.4,2) {};
      \node [circle,draw,fill=white] (i7) at (3.2,2) {};
      \node [circle,draw,fill=white] (i6) at (4.0,2) {};
      \node [circle,draw,fill=white] (i8) at (5.6,2) {};
      \node [circle,draw,fill=white] (q2) at (0,0) {};
      \node [circle,draw,fill=white] (q1) at (0.8,0) {};
      \node [circle,draw,fill=white] (q4) at (1.6,0) {};
      \node [circle,draw,fill=white] (q3) at (2.4,0) {};
      \node [circle,draw,fill=white] (q6) at (3.2,0) {};
      \node [circle,draw,fill=white] (q5) at (4.0,0) {};
      \node [circle,draw,fill=white] (q8) at (4.8,0) {};
      \node [circle,draw,fill=white] (q7) at (5.6,0) {};
      \draw [red,very thick] (i1) -- (q1) node [left,pos=0.25] {$e_1$};
      \draw [red,very thick] (i2) -- (q2) node [right,pos=0.5] {$e_2$};
      \draw [orange] (i3) -- (q3) node [left,pos=0.6] {$e_3$};
      \draw [orange] (i4) -- (q4) node [right,pos=0.5] {$e_4$};
      \draw [purple] (i5) -- (q5) node [right,pos=0.5] {$e_5$};
      \draw [purple] (i6) -- (q6) node [right,pos=0.5] {$e_6$};
      \draw [blue,very thick] (i7) -- (q7) node [right,pos=0.5] {$e_7$};
      \draw [blue,very thick] (i8) -- (q8) node [right,pos=0.25] {$e_8$};
    \end{tikzpicture}
    & \ &
    \begin{tikzpicture}
      \node [circle,draw,fill=white] (i1) at (0,0) {};
      \node [circle,draw,fill=white] (i3) at (0.4,0) {};
      \node [circle,draw,fill=white] (i2) at (1.2,0) {};
      \node [circle,draw,fill=white] (i5) at (1.6,0) {};
      \node [circle,draw,fill=white] (i4) at (2.4,0) {};
      \node [circle,draw,fill=white] (i7) at (3.2,0) {};
      \node [circle,draw,fill=white] (i6) at (4.0,0) {};
      \node [circle,draw,fill=white] (i8) at (5.6,0) {};
      \node [circle,draw,fill=white] (q2) at (0,2) {};
      \node [circle,draw,fill=white] (q1) at (0.8,2) {};
      \node [circle,draw,fill=white] (q4) at (1.6,2) {};
      \node [circle,draw,fill=white] (q3) at (2.4,2) {};
      \node [circle,draw,fill=white] (q6) at (3.2,2) {};
      \node [circle,draw,fill=white] (q5) at (4.0,2) {};
      \node [circle,draw,fill=white] (q8) at (4.8,2) {};
      \node [circle,draw,fill=white] (q7) at (5.6,2) {};
      \draw [red,very thick] (i2) -- (q2) node [left,pos=0.8] {$e_2$};
      \draw [red,very thick] (i1) -- (q1) node [right,pos=0.8] {$e_1$};
      \draw [orange] (i4) -- (q4) node [left,pos=0.8] {$e_4$};
      \draw [orange] (i3) -- (q3) node [right,pos=0.8] {$e_3$};
      \draw [purple] (i6) -- (q6) node [left,pos=0.85] {$e_6$};
      \draw [purple] (i5) -- (q5) node [right,pos=0.85] {$e_5$};
      \draw [blue,very thick] (i8) -- (q8) node [left,pos=0.85] {$e_8$};
      \draw [blue,very thick] (i7) -- (q7) node [right,pos=0.85] {$e_7$};
    \end{tikzpicture}
    \\
    (a) Upper $(X,Y)$-path&&
    (b) Lower $(X,Y)$-path
  \end{tabular}
  \caption{$(X,Y)$-paths in an underlying graph $G$; $X=\{e_1,e_2\}$, $Y=\{e_7,e_8\}$}
  \label{fig:path}
\end{figure}
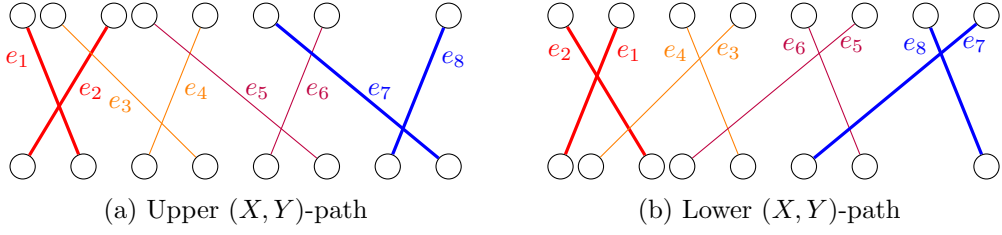

\begin{figure}[t!]
  \centering
    \begin{tikzpicture}
      \node [draw,fill=white] (e1) at (0,0) [] {$e_1$};
      \node [draw,fill=white] (e2) at (1,0) [] {$e_2$};
      \node [draw,fill=white] (e3) at (2,0) [] {$e_3$};
      \node [draw,fill=white] (e4) at (3,0) [] {$e_4$};
      \node [draw,fill=white] (e5) at (4,0) [] {$e_5$};
      \node [draw,fill=white] (e6) at (5,0) [] {$e_6$};
      \node [draw,fill=white] (e7) at (6,0) [] {$e_7$};
      \node [draw,fill=white] (e8) at (7,0) [] {$e_8$};
      \draw [red, very thick] (e1) -- (e2);
      \draw (e2) -- (e3);
      \draw [orange] (e3) -- (e4);
      \draw (e4) -- (e5);
      \draw [purple] (e5) -- (e6);
      \draw (e6) -- (e7);
      \draw [blue, very thick] (e7) -- (e8);
    \end{tikzpicture}
  \caption{Path $H[M]$ in the auxiliary graph $H$ for $(X,Y)$-paths $M$ in \figref{path}}
  \label{fig:HM}
\end{figure}
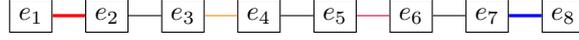

For $X,Y\in\MX$,
we denote by $\rho^\ast(X,Y)$ the max-weight of an $(X,Y)$-path.
We also denote by $\rho^\ast_\odd(X,Y)$ 
(resp., $\rho^\ast_\even(X,Y)$)
the max-weight of an $(X,Y)$-path 
such that the size is odd (resp., even).
We let $\rho^\ast(X,Y)$, $\rho^\ast_\odd(X,Y)$, and $\rho^\ast_\even(X,Y)$
be $-\infty$ when no corresponding path exists.
Clearly we have
\begin{align}
  \rho^\ast(X,Y)=\max\{\rho^\ast_\odd(X,Y),\rho^\ast_\even(X,Y)\}.
  \label{eq:ast}
\end{align}
If $X=Y$, then the path size is two and $\rho^\ast(X,Y)=w(X)$ holds. 
If $X\cap Y=\{e\}$, then the path size is three
and $\rho^\ast(X,Y)=w(X)+w(Y)-w(e)$ holds. 
If $X\cap Y=\emptyset$, then the path size is no less than four. 

%
%

\paragraph{A max-weighted even-sized $(X,Y)$-path.}

We study how to obtain a max-weighted even-sized $(X,Y)$-path for given $X,Y\in\MX$.
We design an algorithm that computes $\rho^\ast_\even(X,Y)$
and constructs the path.  
This strategy is then extended to the odd-size case.

For $X,Y\in\MX$, let $X=(e_X,e'_X)$ and $Y=(e_Y,e'_Y)$. 
We say that an ordered pair $(X,Y)$ is a {\em link}
if one of the followings holds: 
\begin{description}
\item[(a)] $e_X\prec e_Y$,
  $e_X\prec e'_Y$,
  $\{e'_X,e_Y\}\in\MX$, and
  $e'_X\prec e'_Y$. 
\item[(b)] $e_X\prec e_Y$, $\{e_X,e'_Y\}\in\MX$, $e'_X\prec e_Y$, and $e'_X\prec e'_Y$.
\end{description}
If $(X,Y)$ is a link
that satisfies the condition (a) (resp., (b)),
then we call it an {\em upper link}
(resp., a {\em lower link}).
In \figref{path}~(a),
every $(\{e_t,e_{t+1}\},\{e_{t+2},e_{t+3}\})$
with $t\in\{1,3,5\}$ is an upper link. 
In \figref{path}~(b),
every $(\{e_t,e_{t+1}\},\{e_{t+2},e_{t+3}\})$
with $t\in\{1,3,5\}$ is a lower link. 

The following lemma gives a characterization of an even-sized path.


\begin{lem}
  \label{lem:link}
  For 
  a $2$-layered bipartite graph $G$ and an admissible set $\MX$,
  let $M\in\MM_2$ such that $M=\{e_1,\dots,e_{2d}\}$ for an integer $d\ge2$.
  If $H[M]$ is a path  
  $e_1\rightarrow\dots\rightarrow e_{2d}$ such that $e_1<_Ae_{2d}$, 
  then all of $(\{e_{2t-1},e_{2t}\},\{e_{2t+1},e_{2t+2}\})$,
  $t\in[d-1]$,
  are either upper links or lower links. 
\end{lem}
\begin{proof}
  For $t\in[d]$, we denote $X_t=\{e_{2t-1},e_{2t}\}$.
  For $s,s'\in[2d]$,
  if $|s-s'|=1$, then $\{e_s,e_{s'}\}\in\MXall[M]$ holds,
  which means that two edges $e_s$ and $e_{s'}$ intersect.
  Otherwise, since $\{e_s,e_{s'}\}$ does not belong to $\MXall$, 
  they are comparable (\lemref{either}).
  Among the edges in $X_t$ and $X_{t+1}$, $t\in[d-1]$,
  $e_{2t}\in X_t$ and $e_{2t+1}\in X_{t+1}$ intersect,
  and $\{e_{2t-1},e_{2t+1}\}$, $\{e_{2t-1},e_{2t+2}\}$ and $\{e_{2t},e_{2t+2}\}$
  are comparable pairs.  
  We claim that $e_{2t-1}\prec e_{2t+1}$ should hold;
  if not so, let $t'\in[d-1]$ denote the smallest index such that $e_{2t'+1}\prec e_{2t'-1}$.
  This means $e_1\prec\dots\prec e_{2t'-1}$. 
  We have $e_1\prec e_{2d}$ by assumption.
  If $e_{2d}\prec e_{2t'-1}$, 
  there is an edge $e_s$ $(s\in[2,2t'-2])$ that intersects with $e_{2d}$,
  which contradicts that $e_{2d}$ intersects with only $e_{2d-1}$. 
  If $e_{2t'-1}\prec e_{2d}$, since $e_{2t'+1}\prec e_{2t'-1}$,
  there is an edge $e_s$ $(s\in[2t'+2,2d-1])$ that intersects with $e_{2t'-1}$,
  which contradicts that $e_{2t'-1}$ intersects with only $e_{2t'-2}$ and $e_{2t'}$. 
  We can show that $e_{2t-1}\prec e_{2t+2}$ and $e_{2t}\prec e_{2t+2}$ also hold
  in the same way. 

  Now we have two cases: $e_{2t-1}<_Ae_{2t}$ or $e_{2t}<_Ae_{2t-1}$.
  In the former case, $(X_t,X_{t+1})$ is an upper link
  since the condition (a) holds by
  $e_X=e_{2t-1}$, $e'_X=e_{2t}$, $e_Y=e_{2t+1}$ and $e'_Y=e_{2t+2}$.
  Then $e_{2t+1}<_A e_{2t+2}$ also holds.
  For $t\le d-2$, $(X_{t+1},X_{t+2})$ is also an upper link
  since the condition (a) holds by
  $e_X=e_{2t+1}$, $e'_X=e_{2t+2}$, $e_Y=e_{2t+3}$ and $e'_Y=e_{2t+4}$.
  We see that, if $(X_1,X_2)$ is an upper link, then $(X_{t+1},X_{t+2})$ is also an upper link for $t\in[d-2]$. 
  The latter case (i.e., $e_{2t}<_Ae_{2t-1}$) is analogous. 
\end{proof}

For $X,Y\in\MX$,
let $M$ denote an even-sized $(X,Y)$-path
such that the path $H[M]$ is given by $e_1\rightarrow\dots\rightarrow e_{2d}$
and $e_1<_Ae_{2d}$. 
We call $M$ an {\em upper} (resp., a {\em lower}) {\em $(X,Y)$-path}
if all of $(\{e_{2t-1},e_{2t}\},\{e_{2t+1},e_{2t+2}\})$,
$t\in[d-1]$,
are upper (resp., lower) links. 
\figref{path}~(a) and (b)
illustrate these two types of paths.
For convenience, 
we regard that an $(X,X)$-path is an upper path
as well as a lower path. 
We denote by $\rho^\astup_\even(X,Y)$ (resp., $\rho^\astdown_\even(X,Y)$)
the max-weight of an upper (resp., a lower) $(X,Y)$-path
such that the size is even.
We define
$\rho^\astup_\even(X,Y)\triangleq-\infty$
(resp., $\rho^\astdown_\even(X,Y)\triangleq-\infty$)
if no such $(X,Y)$-path exists.  
Clearly we have
\begin{align}
  \rho^\ast_\even(X,Y)=\max\{\rho^\astup_\even(X,Y),\rho^\astdown_\even(X,Y)\}.\label{eq:rhomax}
\end{align}

Due to symmetry, we focus on even-sized upper paths. 
For $X\in\MX$, we define $\Next(X)\triangleq\{Z\in\MX\mid(X,Z)\textrm{\ is\ an\ upper\ link}\}$. 
For $Y\in\MX$, 
we define $\Prev(Y)\triangleq\{Z\in\MX\mid Y\in\Next(Z)\}$.
Moreover, for $j\in[\lambda_A(Y)+1,\gamma_A(Y)-1]$,
we define the subset $\Prev(Y;j)=\{Z\in \Prev(Y)\mid \gamma_A(Z)=j\}$.
For $X,Y\in\MX$ with $X\ne Y$,
an even-sized upper $(X,Y)$-path 
should contain $Z\in\Prev(Y;j)$ for some $j\in[\lambda_A(Y)+1,\gamma_A(Y)-1]$.
In other words, $(Z,Y)$ is the ``last'' upper link on the $(X,Y)$-path. 
Then we define $\rho^\astup_\even(X,Y;j)$
to be the max-weight of an even-sized upper $(X,Y)$-path
such that last upper link on the path, say $(Z,Y)$,
satisfies $\gamma_A(Z)\in[\lambda_A(Y)+1,j]$.
Again, if no such path exists,
we define $\rho^\astup_\even(X,Y;j)\triangleq-\infty$. 
Obviously we have
\begin{align*}
  \rho^\astup_\even(X,Y;\lambda_A(Y)+1)\le\dots\le\rho^\astup_\even(X,Y;\gamma_A(Y)-1)=\rho^\astup_\even(X,Y).
\end{align*}

Let $L=(\MX,\ML)$ denote a digraph such that
the admissible set $\MX$ is the node set
and that $\ML=\{(Z,Z')\in\MX\times\MX\mid Z'\in\Next(Z)\}$
is the arc set.
Since $\lambda_A(Z)<\lambda_A(Z')$ holds for any $(Z,Z')\in\ML$,
no cycle exists in $L$, that is, $L$ is a DAG. 
We define the subset $\MX_X\subseteq\MX$ to be
$\MX_X=\{Y\in\MX\mid\textrm{there\ is\ a\ path\ from\ }X\textrm{\ to\ }Y\textrm{\ in\ }L\}$.
Clearly, if there is an even-sized upper $(X,Y)$-path,
then $Y\in\MX_X$ holds, while the converse does not necessarily hold.  
For $Y\in\MX\setminus \MX_X$, we have
\[
\rho^\astup_\even(X,Y;\lambda_A(Y)+1)=\dots=\rho^\astup_\even(X,Y;\gamma_A(Y)-1)=\rho^\astup_\even(X,Y)=-\infty. 
\]
For $Y\in\Next(X)\subseteq\MX_X$,
the only even-sized upper $(X,Y)$-path is $X\cup Y$. 
Then it holds that;
\begin{align}
  \rho^\astup_\even(X,Y;j)=
  \left\{
  \begin{array}{ll}
    -\infty & \textrm{if\ }j\in[\lambda_A(Y)+1,\gamma_A(X)-1],\\
    w(X)+w(Y) & \textrm{if\ }j\in[\gamma_A(X),\gamma_A(Y)-1]. 
  \end{array}
  \right.
  \label{eq:basecase}
\end{align}
The following lemma gives a characterization of $\rho^\astup_\even(X,Y;j)$
for $Y\in\MX_X\setminus\Next(X)$.

\begin{lem}
  \label{lem:recur}
  For a $2$-layered edge-weighted bipartite graph $G$
  and an admissible set $\MX$,
  let $X,Y\in\MX$. 
  If $Y\in\MX_X\setminus(\{X\}\cup\Next(X))$, then for $j=\lambda_A(Y)+1$, 
  \begin{align}
    \rho^\astup_\even(X,Y;j)=\left\{
    \begin{array}{ll}
      \displaystyle w(Y)+\max_{Z\in\Prev(Y;j)}\Big\{\rho^\astup_\even(X,Z;\lambda_A(Y)-1)\Big\}&
      \textrm{if\ }\Prev(Y;j)\ne\emptyset,\\
      -\infty&\textrm{otherwise},
    \end{array}
    \right.
    \label{eq:recur1}
  \end{align}
  and for any $j\in[\lambda_A(Y)+2,\gamma_A(Y)-1]$,
  \begin{align}
    \rho^\astup_\even(X,Y;j)=\left\{
    \begin{array}{ll}
      \displaystyle\max\Big\{\rho^\astup_\even(X,Y;j-1),w(Y)+\max_{Z\in\Prev(Y;j)}\{\rho^\astup_\even(X,Z;\lambda_A(Y)-1)\}\Big\}
      &\textrm{if\ }\Prev(Y;j)\ne\emptyset,\\
      \rho^\astup_\even(X,Y;j-1)&\textrm{otherwise.}
    \end{array}
    \right.
    \label{eq:recur2}
  \end{align}
\end{lem}
\begin{proof}
  Suppose that $j=\lambda_A(Y)+1$. 
  If $\Prev(Y;j)=\emptyset$,
  then there is no even-sized upper $(X,Y)$-path 
  such that the last link $(Z,Y)$ satisfies $\gamma_A(Z)=j$.
  Hence we have $\rho^\astup_\even(X,Y;j)=-\infty$.
  
  We consider the case of $\Prev(Y;j)\ne\emptyset$.
  Suppose that an even-sized upper $(X,Y)$-path exists.
  There is $Z\in\Prev(Y;j)$ such that $(Z,Y)$ is the last link. 
  let $M^\ast_Z$ be a max-weighted $(X,Y)$-path
  among those having $(Z,Y)$ as the last link.  
  We partition $M^\ast_Z$ into $M^\ast_Z=M_Z\cup Y$,
  where  $M_Z$ is a max-weighted even-sized upper $(X,Z)$-path. 
  Since $Z\ne X$ by $Y\in\MX_X\setminus(\{X\}\cup\Next(X))$,
  there is $Q\in\MX$ such that $(Q,Z)$ is the last upper link of $M_Z$ (\figref{link}).
  No edge in $Q$ should intersect with any edge in $Y$,
  and it holds that $\gamma_A(Q)\in[\lambda_A(Z)+1,\lambda_A(Y)-1]$. 
  Hence we have\[
  \rho^\astup_\even(X,Y;j)=\max_{Z\in\Prev(Y;j)}\{w(M^\ast_Z)\}=w(Y)+\max_{Z\in\Prev(Y;j)}\{w(M_Z)\}
  \]
  and $w(M_Z)=\rho^\astup_\even(X,Z;\lambda_A(Y)-1)$.
  We see that Eq.~\eqref{recur1} determines $\rho^\astup_\even(X,Y;j)$ correctly.  
  If no even-sized upper $(X,Y)$-path exists,
  then $\rho^\astup_\even(X,Z;\lambda(Y)-1)=-\infty$ holds
  for all $Z\in\Prev(Y;j)$ by induction. 
  We have $\rho^\astup_\even(X,Y;j)=-\infty$ by \eqref{recur1}.
  
  The proof for $j\in[\lambda_A(Y)+2,\gamma_A(Y)-1]$ is analogous.
\end{proof}

\begin{figure}[t!]
  \centering
  \begin{tikzpicture}
    \node[circle,draw,fill=white] (yA0) at (8,2) [label=above:$\lambda_A(Y)$] {};
    \node[circle,draw,fill=white] (yA1) at (12,2) [label=above:$\gamma_A(Y)$] {};
    \node[circle,draw,fill=white] (yB0) at (10,0) {};
    \node[circle,draw,fill=white] (yB1) at (11,0) {};
    \draw (yA0) -- (yB1);
    \draw (yA1) -- (yB0);
    \node[circle,draw,fill=gray] (zA0) at (4,2) [label=above:$\lambda_A(Z)$] {};
    \node[circle,draw,fill=gray] (zA1) at (10,2) [label=above:$\gamma_A(Z)$] {};
    \node[circle,draw,fill=gray] (zB0) at (6,0) {};
    \node[circle,draw,fill=gray] (zB1) at (7,0) {};
    \draw (zA0) -- (zB1);
    \draw (zA1) -- (zB0);
    \node[circle,draw,fill=black] (qA0) at (2,2) [label=above:$\lambda_A(Q)$] {};
    \node[circle,draw,fill=black] (qA1) at (6,2) [label=above:$\gamma_A(Q)$] {};
    \node[circle,draw,fill=black] (qB0) at (2,0) {};
    \node[circle,draw,fill=black] (qB1) at (3,0) {};
    \draw (qA0) -- (qB1);
    \draw (qA1) -- (qB0);
  \end{tikzpicture}
  \caption{Three admissible pairs $Q$ (black), $Z$ (gray), and $Y$ (white) such that $(Q,Z)$ and $(Z,Y)$ are upper links}
  \label{fig:link}
\end{figure}
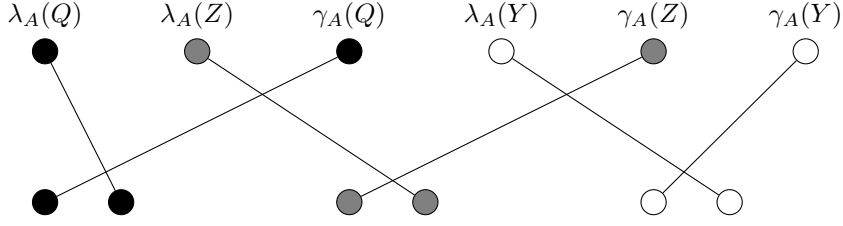

By \lemref{recur}, for a given $X\in\MX$,
we can compute $\rho^\astup_\even(X,Y)$ for all $Y\in\MX_X$ with $Y\ne X$ as follows.
First, we compute the value $\rho^\astup_\even(X,Y;j)$ for all $Y\in \Next(X)$
and $j\in[\lambda_A(Y)+1,\gamma_A(Y)-1]$ by \eqref{basecase}.
Note that $\rho^\astup_\even(X,Y)=\rho^\astup_\even(X,Y;\gamma_A(Y)-1)$ holds. 
Then, if there is $Y\in\MX_X$
such that $\rho^\astup_\even(X,Y)$ has not been determined 
and $\rho^\astup_\even(X,Z)$ has been determined for all $Z\in\Prev(Y)$, 
we compute the value of $\rho^\astup_\even(X,Y;j)$
for all $j\in[\lambda_A(Y)+1,\gamma_A(Y)-1]$ by \eqref{recur1} and \eqref{recur2}.

We summarize the algorithm as
\textsc{EvenUpper} in \algref{evenup}. 
The binary flag $\delta(Y)$ is introduced to
represent whether $\rho^\astup_\even(X,Y)$ has been determined or not. 

\begin{algorithm}[t!]
  \caption{An algorithm \textsc{EvenUpper}
    to compute $\rho^\astup_\even(X,Y)$
    for a given $X\in\MX$ and all $Y\in\MX_X\setminus\{X\}$}
  \label{alg:evenup}
  \DontPrintSemicolon
  \SetKwInOut{Input}{Input}\SetKwInOut{Output}{Output}
  \Input{A $2$-layered bipartite graph $G$,
    a non-empty admissible set $\MX$, and $X\in\MX$}
  \Output{The max-weight $\rho^\astup_\even(X,Y)$
    of an even-sized upper $(X,Y)$-path 
    for all $Y\in\MX_X\setminus\{X\}$}
  \For{$Y\in\MX_X\setminus\{X\}$}{
    \If{$Y\in\Next(X)$}{Compute $\rho^\astup_\even(X,Y;j)$ by \eqref{basecase} for all $j\in[\lambda_A(Y)+1,\gamma_A(Y)-1]$\;
      \label{line:evenup.rho1}
      $\delta(Y)\gets\True$\;
      \label{line:evenup.d1}
    }
    \lElse{$\delta(Y)\gets\False$}
  }
  \While{there is $Y\in\MX_X\setminus\{X\}$ such that $\delta(Y)=\False$ and $\delta(Z)=\True$ for all $Z\in\Prev(Y)$}{
    \label{line:evenup.while}
    Compute $\rho^\astup_\even(X,Y;j)$ by \eqref{recur1} and \eqref{recur2} for all $j\in[\lambda_A(Y)+1,\gamma_A(Y)-1]$\;
    \label{line:evenup.rho2}
    $\delta(Y)\gets\True$\;
    \label{line:evenup.d2}
  }
        {\bf output} $\rho^\astup_\even(X,Y;\gamma_A(Y)-1)$
        as $\rho^\astup_\even(X,Y)$ for all $Y\in\MX_X\setminus\{X\}$
\end{algorithm}

\begin{lem}
  \label{lem:evenup}
  For a $2$-layered edge-weighted bipartite graph $G$,
  a non-empty admissible set $\MX$ and
  an admissible pair $X\in\MX$,
  the algorithm \textsc{EvenUpper} in \algref{evenup}
  computes $\rho^\astup_\even(X,Y)$
  for all $Y\in\MX_X\setminus\{X\}$
  in $O(k^2+kn)$ time and space.
\end{lem}
\begin{proof}
  For $Y\in\Next(X)$, $\rho^\astup_\even(X,Y;j)$ is computed
  in \lineref{evenup.rho1}.  
  We see that the algorithm visits all $Y\in\MX_X\setminus(\{X\}\cup\Next(X))$ by induction with respect to the length of the longest path from $X$.
  When $Y$ is visited,
  $\rho^\astup_\even(X,Z;\lambda_A(Y)-1)$ is already computed
  for all $Z\in\Prev(Y)$.
  Then $\rho^\astup_\even(X,Y;j)$ is computed correctly
  in \lineref{evenup.rho2}
  for all $j\in[\lambda_A(Y)+1,\gamma_A(Y)-1]$ (\lemref{recur}). 
  %
  
  The algorithm \textsc{EvenUpper} can be implemented as follows.
  For preprocessing,
  we construct the DAG $L$
  and the family $\MX_X$, which takes $O(k^2)$ time. 
  During the execution,
  we maintain all $Y$ that satisfy the condition of \lineref{evenup.while}
  in a queue. 
  Every time $\delta(Y)$ is set to $\True$
  (i.e., lines~\ref{line:evenup.d1}
  and \ref{line:evenup.d2}),
  we check whether each $Y'\in\Next(Y)\setminus\Next(X)$
  satisfies the condition of \lineref{evenup.while};
  if yes, we insert $Y'$ to the queue.
  We can do the check over all $Y\in\MX_X\setminus\{X\}$ in $O(k^2)$ amortized time
  since $\sum_{Y}|\Next(Y)|=O(k^2)$.
  For $Y\in\MX_X\setminus(\{X\}\cup\Next(X))$ and $j\in[\lambda_A(Y)+1,\gamma_A(Y)-1]$,
  we can compute $\rho^\astup_\even(X,Y;j)$ in $O(|\Prev(X,Y;j)|)$ time.
  We obtain $\rho^\ast_\even(X,Y)$ by computing $\rho^\astup_\even(X,Y;j)$ for all $j$,
  which takes $\sum_jO(|\Prev(X,Y;j)|)=O(|\Prev(X,Y)|)$ time. 
  Therefore, we can compute $\rho^\ast_\even(X,Y)$ for all $Y$ in $O(k^2+kn)$ time
  since every $Y$ is inserted to the queue exactly once,
  $\sum_YO(|\Prev(X,Y)|)=O(k^2)$, and
  there are $O(kn)$ entries for $\rho^\ast_\even(X,Y;j)$. 
  For the space complexity,
  we use $O(k^2)$ space for $L$,
  and $O(kn)$ space for
  all $\rho^\astup_\even(X,Y;j)$ and $\delta(Y)$.
\end{proof}

The algorithm can be used to construct
a max-weighted even-sized upper $(X,Y)$-path. 

\begin{lem}
  \label{lem:evenup_path}
  For a $2$-layered edge-weighted bipartite graph $G$,
  a non-empty admissible set $\MX$ and
  an admissible pair $X\in\MX$,
  we can construct a max-weighted even-sized upper $(X,Y)$-path
  for all $Y\in\MX$ in $O(k^2+kn)$ time and space.
\end{lem}
\begin{proof}
  For $Y\in\MX_X\setminus(\{X\}\cup\Next(X))$ and $j\in[\lambda_A(Y)+1,\gamma_A(Y)-1]$,
  let us store the maximizer $Z=Z^\ast$ in \eqref{recur1} and \eqref{recur2} as $\chi(Y;j)$. 
  Specifically, for $j=\lambda_A(Y)+1$,
  we let $\chi(Y;j)\gets Z^\ast$ if $\Prev(Y;j)\ne\emptyset$,
  and otherwise, we let $\chi(Y;j)\gets\Null$.
  For other $j$, 
  we let $\chi(Y;j)\gets Z^\ast$ if $\Prev(Y;j)\ne\emptyset$
  and $\rho^\astup_\even(X,Y;j-1)<w(Y)+\rho^\astup_\even(X,Z^\ast;\lambda_A(Y)-1)$, and otherwise, we let $\chi(Y;j)\gets\chi(Y;j-1)$.
  Observe that, if $Z=\chi(Y;\gamma_A(Y)-1)$ is not $\Null$,
  then $(Z,Y)$ 
  is the last link of a max-weighted even-sized upper
  $(X,Y)$-path, and otherwise, no such path exists. 
  
  Upon completion of the algorithm,
  if $Y=X$, the only $(X,Y)$-path is $X$ itself. 
  If $Y\in\Next(X)$, then $X\cup Y$ is the required path.
  If $Y\in\MX_X\setminus(\{X\}\cup\Next(X))$
  and $\rho^\astup_\even(X,Y)\ne-\infty$, then
  we can construct a max-weighted even-sized
  upper $(X,Y)$-path by tracing $\chi(Y,\gamma_A(Y)-1)$,
  which requires $O(n)$ time
  since the path size is at most $n$.
  For all the other $Y$, there is no even-sized upper $(X,Y)$-path.
  The algorithm runs in $O(k^2+kn)$ time and space
  (\lemref{evenup}). 
  The construction of paths can be done in $O(kn)$ time.
  We can store $\chi(Y;j)$ in $O(kn)$ space. 
\end{proof}

We can derive the similar results for even-sized lower paths
due to symmetry. 

\begin{lem}
  \label{lem:evendown}
  For a $2$-layered edge-weighted bipartite graph $G$,
  a non-empty admissible set $\MX$ and
  an admissible pair $X\in\MX$,
  we can construct a max-weighted even-sized lower $(X,Y)$-path
  for all $Y\in\MX$ in $O(k^2+kn)$ time and space.
\end{lem}

\begin{lem}
  \label{lem:even}
  For a $2$-layered edge-weighted bipartite graph $G$,
  a non-empty admissible set $\MX$, 
  we can compute 
  a max-weighted even-sized $(X,Y)$-path
  and its weight $\rho^\ast_\even(X,Y)$
  for all $X,Y\in\MX$
  in $O(k^3+k^2n)$ time and in $O(k^2+kn)$ space.
\end{lem}

\begin{proof}
  By Lemmas~\ref{lem:evenup},
  \ref{lem:evenup_path} and \ref{lem:evendown},
  given $X\in\MX$, we can compute
  max-weighted even-sized upper and lower $(X,Y)$-paths
  and their weights (i.e., $\rho^\astup_\even(X,Y)$ and $\rho^\astdown_\even(X,Y)$) for all $Y\in\MX$
  in $O(k^2+kn)$ time and space.
  The path having a larger weight
  is the required path by \eqref{rhomax}.
  Since $|\MX|=k$, we have the time complexity
  $O(k(k^2+kn))=O(k^3+k^2n)$.
\end{proof}

\paragraph{A max-weighted odd-sized $(X,Y)$-path.}
We compute a max-weighted odd-sized $(X,Y)$-path for given $X,Y\in\MX$
by extending the strategy for the even case that we explained so far.

We consider how to compute the max-weight $\rho^\ast_\odd(X,Y)$. 
Observe that the minimum odd size of an $(X,Y)$-path is three. 
Let $X=(e_X,e'_X)\in\MX$ and $Z=(e_Z,e'_Z)\in\MX$. 
We say that an ordered pair $(X,Z)$ is a {\em wedge}
if one of the followings holds:
\begin{description}
\item[(a)] $e_X=e_Z$ and $e'_X\prec e'_Z$.
\item[(b)] $e_X\prec e_Z$ and $e'_X=e'_Z$. 
\end{description}
If $(X,Z)$ is a wedge that satisfies (a) (resp., (b)),
then we call it an {\em upper wedge} (resp., a {\em lower wedge}). 
See \figref{wedge}.
Analogously to \lemref{link},
we can show that 
any larger odd-sized path is obtained by connecting
upper links to an upper wedge
or lower links to a lower wedge. 

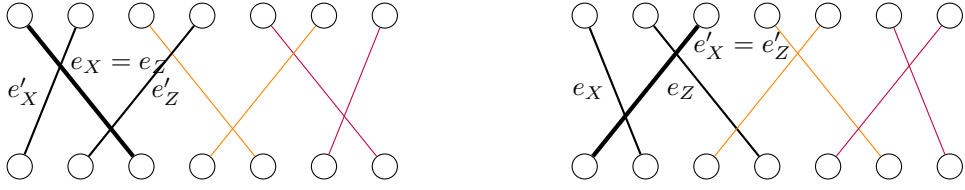
\begin{figure}[t!]
  \centering
  \begin{tabular}{clc}
    \begin{tikzpicture}
      \node [circle,draw,fill=white] (i0) at (0,2) {};
      \node [circle,draw,fill=white] (i1) at (0.8,2) {};
      \node [circle,draw,fill=white] (i2) at (1.6,2) {};
      \node [circle,draw,fill=white] (i3) at (2.4,2) {};
      \node [circle,draw,fill=white] (i4) at (3.2,2) {};
      \node [circle,draw,fill=white] (i5) at (4.0,2) {};
      \node [circle,draw,fill=white] (i6) at (4.8,2) {};
      \node [circle,draw,fill=white] (q0) at (0,0) {};
      \node [circle,draw,fill=white] (q1) at (0.8,0) {};
      \node [circle,draw,fill=white] (q2) at (1.6,0) {};
      \node [circle,draw,fill=white] (q3) at (2.4,0) {};
      \node [circle,draw,fill=white] (q4) at (3.2,0) {};
      \node [circle,draw,fill=white] (q5) at (4.0,0) {};
      \node [circle,draw,fill=white] (q6) at (4.8,0) {};
      \draw [ultra thick] (i0) -- (q2) node [right,pos=0.3] {$e_X=e_Z$};
      \draw [thick] (i1) -- (q0) node [left,pos=0.5] {$e'_X$};
      \draw [orange] (i2) -- (q4);
      \draw [thick] (i3) -- (q1) node [right,pos=0.5] {$e'_Z$};
      \draw [purple] (i4) -- (q6);
      \draw [orange] (i5) -- (q3);
      \draw [purple] (i6) -- (q5);
    \end{tikzpicture}
    & \ &
    \begin{tikzpicture}
      \node [circle,draw,fill=white] (i0) at (0,2) {};
      \node [circle,draw,fill=white] (i1) at (0.8,2) {};
      \node [circle,draw,fill=white] (i2) at (1.6,2) {};
      \node [circle,draw,fill=white] (i3) at (2.4,2) {};
      \node [circle,draw,fill=white] (i4) at (3.2,2) {};
      \node [circle,draw,fill=white] (i5) at (4.0,2) {};
      \node [circle,draw,fill=white] (i6) at (4.8,2) {};
      \node [circle,draw,fill=white] (q0) at (0,0) {};
      \node [circle,draw,fill=white] (q1) at (0.8,0) {};
      \node [circle,draw,fill=white] (q2) at (1.6,0) {};
      \node [circle,draw,fill=white] (q3) at (2.4,0) {};
      \node [circle,draw,fill=white] (q4) at (3.2,0) {};
      \node [circle,draw,fill=white] (q5) at (4.0,0) {};
      \node [circle,draw,fill=white] (q6) at (4.8,0) {};
      \draw [thick] (i0) -- (q1) node [left,pos=0.5] {$e_X$};
      \draw [thick] (i1) -- (q3) node [left,pos=0.5] {$e_Z$};
      \draw [ultra thick] (i2) -- (q0) node [right,pos=0.15] {$e'_X=e'_Z$};
      \draw [orange] (i3) -- (q5);
      \draw [orange] (i4) -- (q2);
      \draw [purple] (i5) -- (q6);
      \draw [purple] (i6) -- (q4);
    \end{tikzpicture}
    \\
    (a) An upper wedge followed by upper links &&
    (b) A lower wedge followed by lower links
  \end{tabular}
  \caption{Wedges $(X,Y)$ and odd-sized paths; a thick edge indicates the edge that is contained in both $X$ and $Y$}
  \label{fig:wedge}
\end{figure}

Then for $X,Y\in\MX$,
we can compute $\rho^\ast_\odd(X,Y)$
in a similar fashion to the even case. 
%
%
We denote by $\rho^\astup_\odd(X,Y)$ (resp., $\rho^\astdown_\odd(X,Y)$)
the max-weight of an odd-sized upper (resp., a lower) $(X,Y)$-path.
We define
$\rho^\astup_\odd(X,Y)\triangleq-\infty$
(resp., $\rho^\astdown_\odd(X,Y)\triangleq-\infty$)
if no such path exists.
Clearly we have
\[
  \rho^\ast_\odd(X,Y)=\max\{\rho^\astup_\odd(X,Y),\rho^\astdown_\odd(X,Y)\}.
\]

Let us focus on upper paths.
The analysis is different from the even case
in that wedges should be taken into account
in the current case.
For $X,Z\in\MX$, we define $\Next_X(Z)$ as follows;
\[
\Next_X(Z)\triangleq\left\{
\begin{array}{ll}
  \{Y\in\MX\mid(Z,Y)\textrm{\ is\ an\ upper\ wedge}\} &\textrm{if\ }Z=X,\\
  \{Y\in\MX\mid(Z,Y)\textrm{\ is\ an\ upper\ link}\} &\textrm{otherwise}. 
\end{array}
\right.
\]
For $Y\in\MX$, we define $\Prev_X(Y)\triangleq\{Z\in\MX\mid Y\in\Next_X(Z)\}$
and for $j\in[\lambda_A(Y)+1,\gamma_A(Y)-1]$,
$\Prev_X(Y;j)\triangleq\{Z\in\Prev_X(Y)\mid\gamma_A(Z)=j\}$.
We define 
$\rho^\astup_\odd(X,Y;j)$
to be the max-weight of an odd-sized upper $(X,Y)$-path
such that $\gamma_A(Z)\in[\lambda_A(Y)+1,j]$ holds,
where $Z$ is an admissible pair on the path
that satisfies $Y\in\Next_X(Z)$. 
If no such path exists,
we define $\rho^\astup_\odd(X,Y;j)\triangleq-\infty$.

We consider a digraph $L_X=(\MX,\ML_X)$ such that
the node set is $\MX$ and
the arc set $\ML_X$ is defined by
$\ML_X=\{(Z,Z')\in\MX\times\MX\mid Z'\in\Next_X(Z)\}$.
We see that $L_X$ is a DAG.
Let $\MY_X=\{Y\in\MX\mid\textrm{there\ is\ a\ path\ from\ }X\textrm{\ to\ }Y\textrm{\ in\ }L_X\}$.
Whenever an odd-sized upper $(X,Y)$-path exists,
$Y\in\MY_X$ holds. 
Then we have $\rho^\astup_\odd(X,Y;j)=-\infty$
for $Y\in\MX\setminus\MY_X$
and $j\in[\lambda_A(Y)+1,\gamma_A(Y)-1]$,
and for $Y\in\Next_X(X)\subseteq\MY_X$, 
\[
\rho^\astup_\odd(X,Y;j)=\left\{
\begin{array}{ll}
  -\infty&\textrm{if\ }j\in[\lambda_A(Y)+1,\gamma_A(X)-1],\\
  w(X\cup Y)&\textrm{if\ }j\in[\gamma_A(X),\gamma_A(Y)-1].
\end{array}
\right.
\]
For $Y\in\MY_X\setminus(\{X\}\cup\Next_X(X))$,
we have the following lemma,
which is analogous to \lemref{recur} in the even case. 

\begin{lem}
  \label{lem:recur_odd}
  For a $2$-layered edge-weighted
  bipartite graph $G$ and an admissible set $\MX$,
  let $X,Y\in\MX$. 
  If $Y\in\MY_X\setminus(\{X\}\cup\Next_X(X))$,
  then for $j=\lambda_A(Y)+1$, 
  \begin{align*}
    \rho^\astup_\odd(X,Y;j)=\left\{
    \begin{array}{ll}
      \displaystyle w(Y)+\max_{Z\in\Prev_X(Y;j)}\Big\{\rho^\astup_\odd(X,Z;\lambda_A(Y)-1)\Big\}&
      \textrm{if\ }\Prev_X(Y;j)\ne\emptyset,\\
      -\infty&\textrm{otherwise},
    \end{array}
    \right.
  \end{align*}
  and for any $j\in[\lambda_A(Y)+2,\gamma_A(Y)-1]$,
  \begin{align*}
    \rho^\astup_\odd(X,Y;j)=\left\{
    \begin{array}{ll}
      \displaystyle\max\Big\{\rho^\astup_\odd(X,Y;j-1),w(Y)+\max_{Z\in\Prev_X(Y;j)}\{\rho^\astup_\odd(X,Z;\lambda_A(Y)-1)\}\Big\}
      &\textrm{if\ }\Prev_X(Y;j)\ne\emptyset,\\
      \rho^\astup_\odd(X,Y;j-1)&\textrm{otherwise.}
    \end{array}
    \right.
  \end{align*}
\end{lem}

Given $X$, we can compute $\rho^\astup_\odd(X,Y)$
for all $Y\in\MX$ in $O(k^2+kn)$ time and space,
as we have done for the even case.
Consequently, we have the following lemma. 

\begin{lem}
  \label{lem:odd}
  For a $2$-layered edge-weighted bipartite graph $G$,
  a non-empty admissible set $\MX$, 
  we can compute 
  a max-weighted odd-sized $(X,Y)$-path
  and its weight $\rho^\ast_\odd(X,Y)$
  for all $X,Y\in\MX$
  in $O(k^3+k^2n)$ time and in $O(k^2+kn)$ space.
\end{lem}

\paragraph{Construction of $\MT$.}
Now we are ready to explain
how to construct the collection $\MT$
of trapezoids. 

\begin{lem}
  Given an instance $(G,w,\MX)$ of the MW-$2$-CPEMP,
  we can construct a collection $\MT$ of weighted trapezoids
  in $O(k^3+k^2n+m^2)$ time and in $O(k^2+kn+m^2)$ space
  such that any optimal solution
  in the corresponding NTSP is also optimal
  for $(G,w,\MX)$.   
\end{lem}
\begin{proof}
  It suffices to collect all trapezoids
  that correspond to 2-CPE matchings $M$
  such that $H[M]$ (in the auxiliary graph $H$)
  is an isolated point,
  a cycle or a max-weighted $(X,Y)$-path for some $X,Y\in\MX$.
  It takes $O(m)$ time for isolated points,
  $O(k^2+m^2)$ time for cycles by \lemref{enumcycle},
  and $O(k^3+k^2n)$ time for max-weighted $(X,Y)$-paths
  by Lemmas~\ref{lem:even} and \ref{lem:odd},
  where $\rho^\ast(X,Y)=\max\{\rho^\ast_\odd(X,Y),\rho^\ast_\even(X,Y)\}$ holds by \eqref{ast}.
  The space complexity is immediate from these lemmas. 
\end{proof}

This lemma tells $\Gamma=O(k^3+k^2n+m^2)$. 
By \thmref{NTSP} and $|\MT|=O(k^2+m)$, 
we have the following theorem immediately. 

\begin{thm}
  \label{thm:CPEMP.2}
  Given an instance $(G,w,\MX)$ of the MW-$2$-CPEMP,
  we can find a max-weighted $2$-CPE matching
  in $O\big(k^3+k^2n+m(m+\log n)\big)$ time
  and in $O(k^2+kn+m^2)$ space. 
\end{thm}

\section{Concluding Remarks}
\label{sec:conc}

In the present paper,
we have considered the max-weighted
$c$-CPE matching problem (MW-$c$-CPEMP).
Given $(G,w,\MX)$,
the problem asks to find a max-weighted matching $M^\ast$
such that each edge in $M^\ast$ intersects with at most $c$
other edges in $M^\ast$, and that
all edge crossings are contained in $\MX$.
The MW-$c$-CPEMP is regarded as a relaxation of the MW-0-CPEMP
(i.e., max-weighted crossing-free matching problem).
The degree of relaxation is represented by the constant $c$. 
Besides chronobiology, 
we believe that the MW-$c$-CPEMP should have
many possibilities for application and extension.

Our approach reduces the MW-$c$-CPEMP
to the non-contact trapezoid selection problem (NTSP).
Given $(\MT,\omega)$, the NTSP asks 
to find a max-weighted subcollection of trapezoids
such that no two of them contact each other. 
We propose an algorithm for the NTSP that runs
in $O(|\MT|\log n+n)$ time (\thmref{NTSP}).
The proposed algorithm is an extension of
the Malucelli et al.'s algorithm
for the MW-0-CPEMP~\cite{MOP.1993}.  
We 
construct an NTSP instance $(\MT,\omega)$
such that any optimal solution for $(\MT,\omega)$ is
also optimal for the underlying MW-$c$-CPEMP instance. 
We explained how to construct such an NTSP instance
and showed that the problem is polynomially solvable
for $c\le2$ (Theorems~\ref{thm:CPEMP.0} to \ref{thm:CPEMP.2}).
We leave analyses for the cases of $c\ge3$ open.

The constraint that we have treated in the paper,
at-most-$c$-crossings-per-edge,
depends on the given drawing of the graph. 
We dealt with 2-layered drawing on parallel straight lines,
one of the conventional drawing layouts,
but the problem can be extended to others;
e.g., 2-layered drawing on curves~\cite{GGL.2008},
2-layered radial drawing~\cite{BETT.1999},
2D geometry drawing~\cite{CARB.2015}.
Observe that the drawing implicitly specifies
a conflict list of edge pairs that should not be
included in a solution at the same time. 
Like the recent studies mentioned in \secref{bg.rel},
it would be interesting to explore
problems of finding an ``optimal'' subgraph
under the at-most-$c$-conflict constraint,
where a conflict list 
is given as a part of the input
instead of a graph drawing.

\end{document}